\documentclass[letterpaper, 10 pt, conference]{IEEEconf} 
\IEEEoverridecommandlockouts   
\newtheorem{definition}{Definition}

\newtheorem{assumption}{Assumption}

\newtheorem{lem}{Lemma}
\newtheorem{theorem}{Theorem}
\newtheorem{corollary}{Corollary}

\usepackage{graphicx}  
\usepackage{amsfonts}
\usepackage{amsmath}
\usepackage{array}
\usepackage{algpseudocode}
\usepackage{float}
\usepackage{graphicx}
\usepackage{mathtools}
\usepackage{multirow}
\usepackage{url}

\title{A Fully Homomorphic Encryption Scheme for Real-Time Safe Control}

\author{Pieter Stobbe$^{1}$, Twan Keijzer$^{1}$ and Riccardo M.G. Ferrari$^{1}$  
\thanks{$^{1}$Delft Center for Systems and Control, Delft University of Technology, The Netherlands.
        {\tt\small pj.stobbe@hotmail.com, \{t.keijzer, r.Ferrari\}@tudelft.nl}.}}

\begin{document}

\maketitle

\begin{abstract}  
Fully Homomorphic Encryption (FHE) has made it possible to perform addition and multiplication operations on encrypted data. Using FHE in control thus has the advantage that control effort for a plant can be calculated remotely without ever decrypting the exchanged information. FHE in its current form is however not practically applicable for real-time control as its computational load is very high compared to traditional encryption methods. In this paper a reformulation of the Gentry FHE scheme is proposed and applied on an FPGA to solve this problem. It is shown that the resulting FHE scheme can be implemented for real-time stabilization of an inverted double pendulum using discrete time control.
\end{abstract}

\section{Introduction}\label{sec:intro}
\noindent Cryptography has allowed for the development of control-systems, such as hydroelectric dams or energy grids at a regional level or higher, that must be securely monitored and controlled over long distances. Such spatially distributed systems require a remote connection between the plant actuators, sensors and the controller which can only be feasibly secured from intrusion via encryption. 

Currently, industrial control systems are secured by end-to-end encryption, utilising a mix of symmetric-key and public-key encryption schemes \cite{AES,RSA,EncInCrtl}. These methods are successful in securing low rate communication within large scale control systems. However, they are unsuitable for high sampling frequency feedback-control, as they require multiple encryption and decryption steps. These operations introduce time overhead, reducing the stability margin and possibly de-stabilizing the plant. Furthermore the decryption of data at the controller level means that these methods do not provide security if the controller itself is compromised.

Homomorphic encryption (HME) schemes present a solution to these problems. These schemes allow for multiplication and/or addition of encrypted numbers, thus removing the need for decryption and encryption at the controller level. There are two main types of HME: Partially (PHE) and Fully Homomorphic Encryption (FHE). PHE schemes support only multiplication or addition, whereas FHE schemes support both. The first HME scheme was RSA \cite{RSA}, followed by PHE schemes such as EL Gamal \cite{ElGamal} and Paillier \cite{Paillier}.

More recently lattice-based FHE schemes have been introduced in \cite{Gentry2009,gentry2013,Cheon2015}. For encrypted control this means these schemes allow for implementation of a broad range of feedback control. However, the high computational complexity of these lattice based schemes prevents them from being used in real-time on conventional hardware.

PHE schemes have also been proposed for control, such as in \cite{HME_networked_control} which proposes a combination of the El Gamal \cite{ElGamal} and RSA \cite{RSA} schemes. This control scheme, however, requires the controller state to be sent to the plant for decryption and re-encryption at each time step, adding additional overhead. More recently, \cite{RTS_HME_FPGA} has demonstrated direct feedback control with the PHE scheme from \cite{Paillier}. Due to the limited homomorphic properties of PHE, the controller used un-encrypted controller gains, posing a security risk.

However, recently more attention has been directed to FHE schemes for control, such as in \cite{PeriodicReset,Kim2016,Quant_proj,Kim2021,Chaher2021}. These schemes however still suffer from two problems. First of all, the representation of encrypted signals requires orders of magnitude more storage than the original plaintext. This means that, due to limited computation and bandwidth resources, real-time control with FHE is limited in complexity and update rate. In \cite{Kim2016}, a two-state LTI controller is implemented with an update rate of $2~$Hz while in \cite{Cheon2018} a direct feedback controller for high-level control of a drone reaches an update rate of $10~$Hz.

Secondly, PHE and FHE only allow for encryption of unsigned integers, whereas control requires the use of real numbers. For PHE and FHE the real numbers can be represented as unsigned integers through the Q format. One limitation that remains is that multiplications will shift the location of the decimal point, eventually leading to overflow. Under normal conditions, the decimal point can be shifted back with right hand bit-shifts. However, no FHE schemes currently support homomorphic right hand bit-shifts without excessive penalties on \emph{multiplicative depth}, which is defined as the maximum allowed number of consecutive multiplications. Lattice based encryption schemes support only a relatively shallow multiplicative depth, after which ciphers can no longer be correctly decrypted. Alternative solutions, such as periodic reset \cite{PeriodicReset} and scaling of the state space matrices \cite{Quant_proj} have been proposed to solve this problem. These methods however affect stability and performances, limiting applicability.

The problems of computational complexity and fixed precision have hindered the acceptance of FHE for real-time control. In this paper we propose an FHE scheme for real-time secure control implemented on an Field Programmable Gate Array (FPGA) which address this issue. The contributions of the paper are:
\begin{itemize}
    \item The Gentry's FHE scheme \cite{Gentry2009} has been reformulated with analytical operations that allows for more intuitive manipulation of the scheme.
    \item A so-called \emph{reduced cipher} is introduced via a change of representation of the original cipher to reduce the computational complexity of the FHE scheme.
    \item The FHE scheme is implemented on an FPGA for real-time control of an unstable plant to demonstrate the benefits of the novel \emph{reduced cipher}.
\end{itemize}
The resulting FHE scheme can be used in combination with existing control schemes based on additions and multiplications and, while an FPGA was chosen here, can also be implemented on any conventional hardware. Note that the encryption properties of Gentry's scheme \cite{Gentry2009} are unchanged by using the novel \emph{reduced cipher}.

In the following Section \ref{sec:prob} introduces the considered control setup and Gentry's FHE scheme. In Section \ref{sec:fpga} the \emph{reduced cipher} is introduced and its equivalence is proven. Section \ref{sec:results} shows results of implementing FHE on an FPGA for control of a inverted double pendulum. 

\subsection{Notation} \label{ssec:notation}
\noindent For a positive scalar $x$, we denote individual digits of its binary representation as $x^{[i]}$. That is, $x=\sum_{i=0}^\infty 2^ix^{[i]}$. For any $x\in\mathbb{N}$ we define $(x)^\ell=\sum_{i=0}^{\ell-1} 2^ix^{[i]}$ which are the $\ell$ least significant binary digits of $x$, such that if $x \leq q$ where $q = 2^\ell - 1$, then $(x)^\ell = x$ and if $x > q$, then $(x)^\ell \neq x$. We denote $[x]^\ell=[x^{[0]},\dots,x^{[\ell]}]$ as a vector whose elements are the binary  digits of $(x)^\ell$; $g = [2^0,\dots, 2^{\ell-1}]^\top$ and the set $\mathbb{Z}_q = \{0,\dots,q-1\}$, where $q \in \mathbb{N}$. We denote bit-shifts of a $x$ by $i$ bits as $x\ll i =2^ix$ and  $x\gg i =2^{-i}x$. These concepts can be extended to matrices $X\in \mathbb{N}^{n_1 \times n_2}$, where $(X)^\ell$, $[X]^\ell$, and bitshifts are applied element-wise. $G_{n} = I_{n} \otimes g$, while the encrypted version of a variable $x$ is denoted as $\textbf{E}(x)$. Finally, for a discrete time signal $x(k)$, where $k$ is the current time step, the short-hand notation $x^+ = x(k+1)$ is used.

\section{Problem Statement}\label{sec:prob}
\noindent This section will cover the control scenario in Section \ref{ssec:control_scen}, followed by section \ref{ssec:gentry}, which introduces Gentry's FHE scheme using the proposed novel, simplified notation.

\subsection{Control Scenario}\label{ssec:control_scen}
\noindent In this paper we consider a nonlinear plant of the form 
\begin{equation}\label{eq:generalplant}
    \begin{dcases}
        \dot{x}=f(x,u)+\xi,\\
        y=h(x,u)+\eta,
    \end{dcases}
\end{equation}
and a discrete time, dynamical linear controller of the form
\begin{equation}\label{eq:generalcontroller}
    \begin{dcases}
        \hat{x}^+ = g(\hat{x},u,y,L),\\
        u^+ = v(\hat{x}^+,K),
    \end{dcases}
\end{equation}
where $x \in \mathbb{R}^\rho$ is the state, $\hat{x} \in \mathbb{R}^\rho$ the state estimate, $u \in \mathbb{R}^\gamma$ the input and $y\in \mathbb{R}^\nu$ the output. $f(\cdot)$ and $h(\cdot)$ are the known state transition and output functions, and $\xi$ and $\eta$ represent external disturbances or model uncertainty. The controller consists of two parts: $g(\cdot)$ to obtain the next $\hat{x}$; and $v(\cdot)$ to obtain the new control input. In this paper we consider the plant is controlled by an encrypted version of the controller, which using notation from Section \ref{ssec:notation} is denoted by
\begin{equation}\label{eq:generalencryptedcontroller}
    \begin{dcases}
        \textbf{E}(\hat{x}^+) = \tilde{g}(\textbf{E}(\hat{x}),\textbf{E}(u),\textbf{E}(y),\textbf{E}(L))\,,\\
        \textbf{E}(u^+) = \tilde{v}(\textbf{E}(\hat{x}^+), \textbf{E}(K))\,.
    \end{dcases}
\end{equation}

\noindent The entire encrypted control loop is shown in Figure \ref{fig:setup}, which will be discussed in more detail in Section \ref{fig:setup}. To limit the scope of this paper to its focus of encryption, we make two assumptions on the unencrypted control. 
\begin{assumption}
The control law \eqref{eq:generalcontroller} can be constructed with addition, subtraction and multiplication operations only. This holds true for all linear control methods such as PID, state-feedback and LQR control \cite{control_theory_introduction_book}.
\end{assumption}

\begin{assumption}
The plant in Equation \eqref{eq:generalplant} is stabilised by the unencrypted controller \eqref{eq:generalcontroller}.
\end{assumption}

\subsection{Fully\,Homomorphic\,Encryption\,Scheme\,by\,Gentry\,et~al.}\label{ssec:gentry}
\noindent In this paper, Gentry's FHE scheme \cite{gentry2013} is adapted to become more computationally efficient. Gentry's FHE scheme consists of four procedures: Key generation, encryption, homomorphic operations, and decryption. Gentry introduced four functions to perform these procedures. These functions are defined using notation from Section \ref{ssec:notation} as follows:
\begin{definition}\label{def:funcs}
For any matrix $a\in\mathbb{N}^{N \times (n+1)}$, $b\in\mathbb{N}^{N\times N}$, and $c\in \mathbb{Z}_q^{n+1\times 1}$
\begin{align}
&\textrm{\textbf{BitDecomp}}(a) = [a]^\ell\label{eq:bitdecomp}\\
&\textrm{\textbf{BitDecomp}}^{-1}(b)= b \cdot G_{n+1} \label{eq:bitdecomp-}\\
    &\textrm{\textbf{Flatten}}(b) = [b \cdot G_{n+1} ]^\ell\label{eq:flatten}\\
    & \textrm{\textbf{PowersOf2}}(c) = c \cdot G_{n+1}^\top
\end{align}
\end{definition}
We are now ready to define the procedures used in \cite{gentry2013}.
\subsubsection*{Key Generation}
A public-private key pair is generated as follows: Pick parameters $m\in\mathbb{N}$, $n\in\mathbb{N}$ and $q\in\mathbb{N}$ based on the required security and precision respectively. The private key is $s = \left[1, -t\right]^\top$ where $t \in \mathbb{Z}^{1\times n}_q$ is sampled uniformly on the interval $[0,q-1]$. The public key is $A = \left[b, B \right]$ where $b = B \cdot t^\top + e$, each element of $B \in \mathbb{Z}^{m\times n}_q$ is sampled uniformly on the interval $[0,q-1]$, and each element of $e \in \mathbb{Z}^{m}_q$ is sampled from the $\chi_q$ distribution \cite{LWE_oded_regev}.

\subsubsection*{Encryption}
A message $\mu\in\mathbb{Z}_q$ can be encrypted as a cipher $C \in \mathbb{Z}_2^{N\times N}$ via the following relation
\begin{equation}\label{eq:orig_enc}
    \hspace{-0.15cm}\begin{aligned}
        C = &\textbf{Enc}(\mu) = \textbf{Flatten}(\mu \cdot I_N + \textbf{BitDecomp}(R\cdot A)) = \\
        &[(\mu \cdot I_N + [R\cdot A]^\ell)\cdot G_{n+1}]^\ell\,,
    \end{aligned}
\end{equation}
where $N=\ell(n+1)$ depends on the message size through $\ell=\lfloor\log_2(q)\rfloor + 1$ and each element of $R \in \mathbb{Z}^{N \times m}_2$ is sampled uniformly on the interval $[0,1]$. 
\subsubsection*{Decryption}
Ciphers are decrypted using the $\textbf{MPDec}$ algorithm as proposed in \cite{MPDec}:
\begin{equation}\label{eq:orig_dec}
\mu = \textbf{MPDec}((C\textbf{PowersOf2}(s))^\ell)
\end{equation}
The \textbf{MPDec} algorithm \cite{MPDec} uses the first $\ell$ elements of its input to retrieve $\mu$. 
Proof that the correct message is retrieved in this way can be found in \cite{gentry2013}.
\subsubsection*{Homomorphic Operations}
The homomorphic operations for ciphers $C_1=\textbf{Enc}(\mu_1)$ , $C_2=\textbf{Enc}(\mu_2)$ and scalar $\alpha$ are
\begin{equation}\label{eq:operation_types}
\begin{aligned}
    \textrm{Sum: } C_3 = &\textbf{Flatten}(C_1 + C_2) = \\ &[(C_1 + C_2) \cdot G_{n+1}]^\ell,\\
    \textrm{Product: } C_4 = &\textbf{Flatten}(C_1 \cdot C_2) =\\ 
    &[(C_1 \cdot C_2) \cdot G_{n+1}]^\ell,\\
    \textrm{Scalar product: } C_5 =& \textbf{Flatten}(\textbf{Flatten}(\alpha I_N) \cdot C_2) =\\ &[([(\alpha I_N) \cdot G_{n+1}]^\ell C_2) \cdot G_{n+1}]^\ell,\\
    \textrm{Scalar sum: } C_6 =& \textbf{Flatten}(\alpha I_N + C_2)= \\
    &[(\alpha I_N + C_2) \cdot G_{n+1}]^\ell.\\
\end{aligned}
\end{equation}
For these homomorphic operations it is proven that
\begin{equation}\label{eq:hom_oper}
    \begin{aligned}
        \mu_3 =&\mu_1 + \mu_2 &\Longleftrightarrow& \; \mu_3 =& \textbf{MPDec}(C_3)\,,\\
        \mu_4 =&\mu_1\mu_2 &\Longleftrightarrow& \; \mu_4 =& \textbf{MPDec}(C_4)\,,\\
        \mu_5 =&\alpha\mu_2 &\Longleftrightarrow& \; \mu_5 =& \textbf{MPDec}(C_5)\,,\\
        \mu_6 =&\alpha+\mu_2 &\Longleftrightarrow& \; \mu_6  =& \textbf{MPDec}(C_6)\,.\\
    \end{aligned}
\end{equation}

\subsection{FHE in Control}\label{ssec:FHE_in_control}
\noindent The Gentry FHE scheme \cite{gentry2013} has excellent theoretical properties, but there are two obstacles which, until now, have prevented implementation of the scheme in control. Firstly, any message $\mu\in \mathbb{Z}_q$ containing $\ell$ bits of information, when encrypted, becomes a cipher $C\in \mathbb{Z}^{N\times N}_2$ containing $N^2=(n+1)^2\ell^2$ bits of information. Therefore storage and transfer of ciphers requires more memory than unencrypted equivalents. The problem of size becomes even more pronounced when performing homomorphic operations. Direct implementation of homomorphic operations requires multiple steps in which intermediate ciphers can become as large as $\mathbb{Z}^{N \times N}_N$ containing $N^2(\lfloor\log_2(N)\rfloor+1)$ bits of information. 

Even more important than the strain on storage and communication, is the strain on the computational resources. For direct implementation of homomorphic addition, $N^2(n+2)-N(n+1)$ addition operations and $N^2(n+1)$ multiplication operations are needed, whereas its unencrypted equivalent requires only a single addition. In this paper a so-called \emph{reduced cipher} is introduced to reduce the computational load of FHE, allowing for faster update rates of control laws.

The second obstacle is the representation of real numbers with unsigned integers. To this end we employ the commonly used fixed precision representation called $Q$ format \cite{RTS_HME_FPGA}\footnote{We will be using the Q-notation as introduced by Texas-Instruments, which is used in code libraries such as the TMS320C64x+ IQmath.}. 
Alternatives using floating point numbers are currently being researched \cite{FHE_Float} but are not yet sufficiently mature. $Q$ format allows for representing a fixed accuracy number $\beta$ with an integer message $\mu\in\mathbb{Z}_p$ where
$\lfloor\log_2(p)\rfloor + 1 = m_q + n_q$ as
\begin{equation}\label{eq:qmn}
\begin{aligned}
    &\beta = -2^{m_q-1}\mu^{[m_q+n_q-1]} + \sum_{i=0}^{m_q+n_q-2} 2^{i-n_q} \mu^{[i]}\\
    &\mu = \left\{\begin{matrix}
    2^{n_q}\beta &\text{ if } \beta \geq 0\\
   -2^{m_q+n_q}+\lvert \beta \rvert 2^{n_q} &\text{ if } \beta <0
    \end{matrix}\right.
\end{aligned}
\end{equation}
such that $\beta$ can be any value in $[-2^{m_q-1},2^{m_q-1})$ rounded to the nearest $2^{-n_q}$. When performing multiplication of two messages $\mu_3= \mu_1\cdot \mu_2$, where $\mu_1$ and $\mu_2$ are obtained from Equation \eqref{eq:qmn}, the result has to fit a $m_q+2n_q$ sized register to yield an exact result. The available storage for each message is limited and so after a certain number of consecutive multiplications overflow would occur.

Therefore, conventionally, a right-bitshift by $n_q$ bits is performed after each multiplication such that the $m_q+n_q$ least significant bits of $\mu_3$ can be used to retrieve $\beta_1\beta_2$\footnote{rounded down to the nearest $2^{-n_q}$, due to truncation during the right hand bitshift.}. However, no HME scheme supports such operation on ciphers without penalty on multiplicative depth. Thus, consecutive multiplications have formed a great obstacle in HME. This problem is important for controllers, which often have internal states that are updated at each timestep without being decrypted. Until now this obstacle has been dealt with using a periodic reset \cite{RTS_HME_FPGA} or by transforming the state space variables \cite{Quant_proj}. These methods, however, affect the stability and performance of the controller such that direct implementation of FHE with existing control schemes is not possible.

\section{Reduced Ciphers for Fast FHE Implementation}\label{sec:fpga}
\noindent In this section the so-called \emph{reduced cipher} will be presented for computationally efficient implementation of FHE for discrete control. It will be shown that, with the \textit{reduced cipher}, encryption, homomorphic operations, and decryption can be made orders of magnitude more computationally efficient, enabling real-time implementation of FHE for control.

Given a cipher $C\in \mathbb{Z}^{N\times N}_2$, the so-called \textit{reduced cipher} will be denoted as $\tilde{C}\in \mathbb{Z}^{N\times (n+1)}_q$ and is defined as
\begin{equation*}
    \tilde{C} = \textbf{BitDecomp}^{-1}(C) = CG_{n+1}\,.
\end{equation*}
Here Definition \ref{def:funcs} is used to rewrite the relation between cipher and \textit{reduced cipher}. Note that the reduced \textit{cipher contains} exactly the same information as the original cipher. In Theorem \ref{thm:red_cipher} it will be shown that using the \textit{reduced cipher} reduces the total number of computer operations needed, and completely eliminates the need for doing hardware multiplications when performing homomorphic multiplication.

\begin{lem}\label{lem:red_cipher}
For any matrix $\Lambda\in\mathbb{N}^{n_1\times n_2}$, we have $\left[\Lambda\right]^\ell G_{n_2}=\left(\Lambda\right)^\ell$.
\end{lem}
\begin{proof}
First consider $\alpha \in \mathbb{N}$. for any $\alpha$ it holds
\begin{equation*}
    (\alpha)^\ell = \sum_{i=0}^{\ell-1} 2^i\alpha^{[i]} = \left[\alpha^{[0]}, \dots, \alpha^{[\ell-1]} \right] \cdot g = \left[ \alpha \right]^\ell \cdot g\,.
\end{equation*}
Then apply this relation on each element of $\Lambda$, giving \\
$ (\Lambda)^\ell = \left[\Lambda \right]^\ell \cdot I_{n_2} \otimes g = \left[\Lambda \right]^\ell \cdot G_{n_2}$
\end{proof}

\begin{theorem}\label{thm:red_cipher}
Given ciphers $C_1,C_2\in\mathbb{Z}^{N\times N}_2$ and scalar $\alpha\in\mathbb{Z}^q$ the existing homomorphic operations can equivalently be written using the \textit{reduced cipher} as
    \begin{align}
        &C_3=[(C_1+C_2)G_{n+1}]^\ell \leftrightarrow \tilde{C}_3 = \left(\tilde{C}_1+\tilde{C}_2\right)^\ell\label{eq:hom_add}\\
        &C_4=[(C_1\cdot C_2)G_{n+1}]^\ell \leftrightarrow \tilde{C}_4 = \left(C_1\cdot \tilde{C}_2\right)^\ell\label{eq:hom_mult}\\
        &\hspace{-0.06cm}C_5=[[\alpha G_{n+1}]^\ell\cdot C_1G_{n+1}]^\ell \leftrightarrow \tilde{C}_5 = (\left[\alpha G_{n+1}\right]^\ell \tilde{C}_1)^\ell\label{eq:hom_multc}\\
        &C_6=[(\alpha I_N+ C_1)G_{n+1}]^\ell \leftrightarrow \tilde{C}_6 = (\alpha G_{n+1} + \tilde{C}_1)^\ell\label{eq:hom_addc}
    \end{align}
\end{theorem}
\begin{proof}
Each equivalence is proven separately below.
\begin{equation*}
    \begin{aligned}
    \tilde{C}_3 = 
        &\left[(C_1 + C_2) G_{n+1}\right]^\ell G_{n+1} \\
        =&\left(C_1 G_{n+1} + C_2 G_{n+1}\right)^\ell
        =\left(\tilde{C}_1 + \tilde{C}_2\right)^\ell\\
    \tilde{C}_4 =
        &\left[(C_1 \cdot C_2) G_{n+1}\right]^\ell G_{n+1} =
        \left(C_1 \cdot \tilde{C}_2\right)^\ell \\
        \tilde{C}_5 =
        &\left[\left[\alpha I_N G_{n+1}\right]^\ell \cdot C_1 G_{n+1}\right]^\ell G_{n+1} \\
        =&\left(\left[\alpha G_{n+1} \right]^\ell \cdot \tilde{C}_1\right)^\ell
        \end{aligned}
        \end{equation*}
        \begin{equation*}
    \begin{aligned}
    \tilde{C}_6 = 
        &\left[(\alpha I_N+  C_1) G_{n+1}\right]^\ell G_{n+1}= \\
        &\left(\alpha G_{n+1} + C_1 G_{n+1}\right)^\ell =
        \left(\alpha G_{n+1} + \tilde{C}_1 \right)^\ell
    \end{aligned}
\end{equation*}
Here Definition \ref{def:funcs} and Lemma \ref{lem:red_cipher} were used.
\end{proof}
Theorem \ref{thm:red_cipher} has shown equivalences between homomorphic operations on ciphers and on \textit{reduced ciphers}. In the following corollaries it will be shown how these equivalences are used in encryption and decryption for the FHE scheme.
\begin{corollary}\label{cor:gamma}
The term $\alpha G_{n+1}$ from Theorem \ref{thm:red_cipher} can be generated using only bitshifts. Due to the structure of $\alpha G_{n+1}$ the number of operations needed to obtain $\tilde{C}_5$ and $\tilde{C}_6$ can be reduced to respectively $\mathcal{O}(n^2\ell^2)$ and $\mathcal{O}(n\ell)$.
\end{corollary}
\begin{proof}
Denote $\alpha G_{n+1}=I_{n+1}\otimes \alpha g$. Note that $\alpha G_{n+1}$ contains $(n+1)$ instances of $\alpha g$, such that only $N=(n+1)\ell$ entries are non-zero. Therefore, by skipping the structural zeros, we require only $\mathcal{O}(N^2)$, $\mathcal{O}(N)$ operations respectively to obtain $\tilde{C}_5$ and $\tilde{C}_6$. Furthermore, $\alpha g$ can be generated using $\mathcal{O}(\ell)$ bitshifts as $\alpha g=[\alpha, \alpha \ll 1,\dots,\alpha \ll \ell-1]$.
\end{proof}
\begin{corollary}\label{cor:red_cipher}
Encryption and decryption can be rewritten in terms of \textit{reduced ciphers} using theorem \ref{thm:red_cipher}.
\end{corollary}
\begin{proof}
Encryption is performed using Equation \eqref{eq:orig_enc}, where $\textbf{BitDecomp}(RA)=[RA]^\ell\in\mathbb{Z}^{N\times N}_2$ is of the same form as a cipher. Encryption is thus a special case of the homomorphic scalar sum as defined in equation \eqref{eq:hom_oper}. Applying Equation \eqref{eq:hom_addc} and Lemma \ref{lem:red_cipher} to encryption yields
\begin{equation}\label{eq:new_enc}
    \tilde{C} = (\mu G_{n+1} + [R\cdot A]^\ell G_{n+1})^\ell = (\mu G_{n+1} + R\cdot A)^\ell
\end{equation}
Decryption can be rewritten using the novel notation as
$\mu = \textbf{MPDec}((CG_{n+1}s)^\ell)=\textbf{MPDec}((\tilde{C}s)^\ell)$
\end{proof}

\begin{table}[h]
    \centering
    \caption{Number of operations required for homomorphic operations.}
    \begin{tabular}{|m{1.6cm}|>{\centering}m{1cm}|>{\centering}m{1cm}|>{\centering}m{1.9cm}|>{\centering\arraybackslash}m{1cm}|}
    \hline
         \multirow{2}*{}& \multicolumn{2}{c|}{sum} & \multicolumn{2}{c|}{product}\\\cline{2-5}
         & Cipher & Red. Cipher & Cipher & Red. Cipher\\\hline
        Bit Operation & $ 0$ &$0$ & $\mathcal{O}(n^3\ell^3)$ &$\mathcal{O}(n^3\ell^2)$\\
        Addition & $\mathcal{O}(n^3\ell^2)$ &$\mathcal{O}(n^2\ell)$ & $\mathcal{O}(n^3\ell^3)$ &$\mathcal{O}(n^3\ell^2)$\\
        Multiplication & $\mathcal{O}(n^3\ell^2)$ &$0$ &$\mathcal{O}(n^3\ell^2)$ & $0$ \\
        Memory & $\mathcal{O}(n^2\ell^2)$ &$\mathcal{O}(n^2\ell^2)$ &$\mathcal{O}(n^2\ell^2\log(n\ell))$ & $\mathcal{O}(n^2\ell^2)$ \\\hline
        & \multicolumn{2}{c|}{scalar sum} & \multicolumn{2}{c|}{scalar product}\\\hline
        Bit Operation & $\mathcal{O}(n^3\ell^2)$ &$\mathcal{O}(\ell)$ & $\mathcal{O}(n^2\ell^3)$ &$\mathcal{O}(n^2\ell^2)$\\
        Addition & $\mathcal{O}(n^3\ell^2)$ &$\mathcal{O}(n\ell)$ & $\mathcal{O}(n^3\ell^3)$ &$\mathcal{O}(n^2\ell^2)$\\
        Multiplication & $\mathcal{O}(n^2\ell)$ &$0$ &$\mathcal{O}(n^3\ell^2)$ & $0$ \\
        Memory & $\mathcal{O}(n^2\ell^2)$ &$\mathcal{O}(n^2\ell^2)$ &$\mathcal{O}(n^2\ell^2\log(\ell))$ & $\mathcal{O}(n^2\ell^2)$ \\\hline
    \end{tabular}
    \label{tab:hom_oper}
\end{table}

The actual reduction of computational complexity obtained by using the \textit{reduced ciphers} for the homomorphic operations is summarized in Table \ref{tab:hom_oper}. The table shows the computational complexity and memory utilisation of the operations that are involved in evaluating the homomorphic operations from equation \eqref{eq:operation_types}, both with an without the use of \textit{reduced ciphers}.
The number of operations is reduced and homomorphic multiplication no longer requires multiplication of cipher elements. Furthermore, Table \ref{tab:hom_oper} shows the reduction in required memory for performing the operations without requiring intermediate reading and writing of memory.
Note however, that the reduced cipher only increases computational efficiency, but does not affect the encryption properties of Gentry's scheme. Furthermore, \textit{Reduced ciphers} contain the same amount of data as regular ciphers and so the communicational bandwidth required to transfer the ciphers is unchanged.

\section{Results on a Simulated Plant}\label{sec:results}
\begin{figure}[h]
    \centering
    \includegraphics[width = 1.04\linewidth]{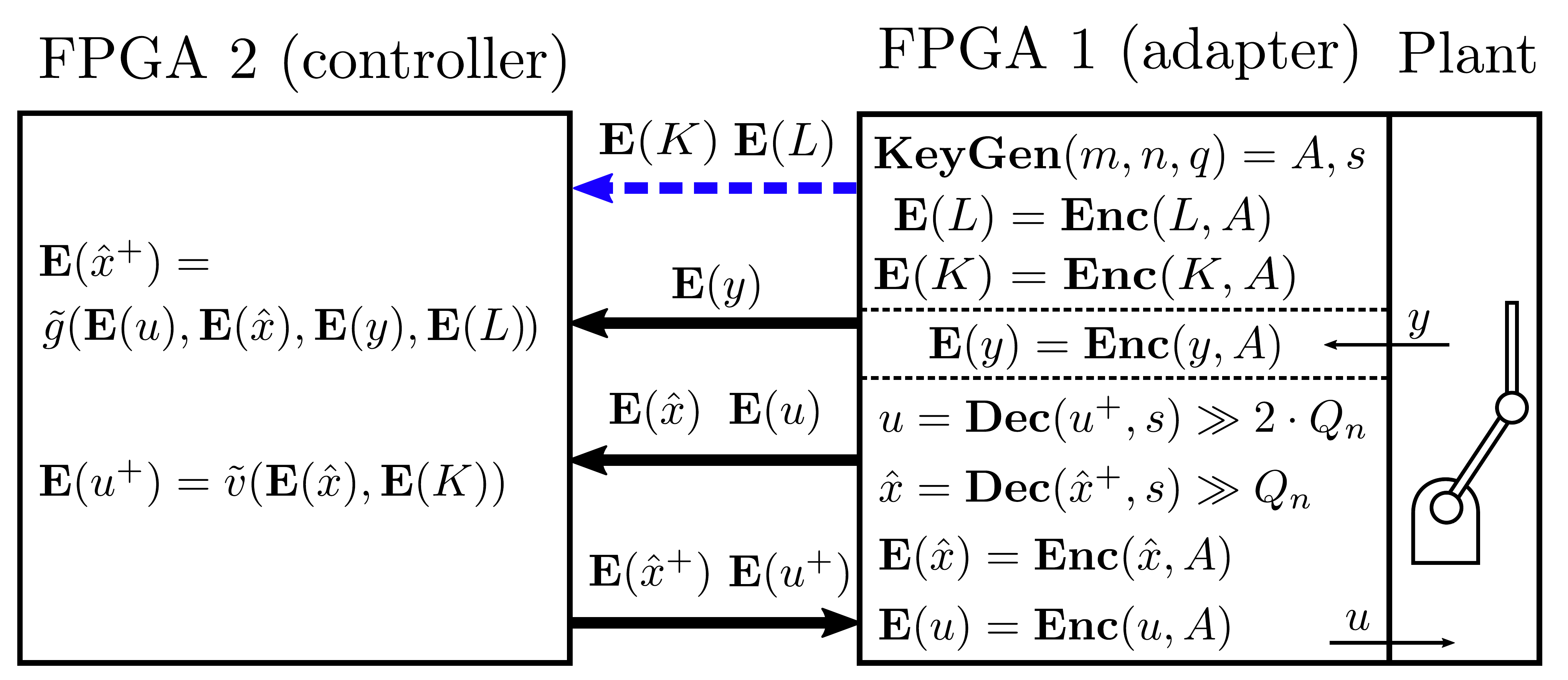}
    \caption{The experimental setup where FPGA 1 performs encryption and decryption at the plant and FPGA 2 contains the remote controller. The key exchange, as indicated with the dashed blue arrow, is only required at initialization.}
    \label{fig:setup}
\end{figure}

\noindent In this section we will apply the novel FHE scheme to the control of an inverted double pendulum. To achieve a realistic setup, the encrypted control is implemented on two FPGAs, as shown in Figure \ref{fig:setup}. It will be shown that it is possible to perform stabilising control of the unstable plant in real-time with the encrypted controller. Below, first the properties of an FPGA and the used setup will be discussed. Next the double pendulum model and control law will be introduced. Lastly, the obtained results will be shown.

\subsection{Hardware Resources of an FPGA}\label{ssec:fpga_hard}
\noindent An FPGA contains generic logic cells and memory components of differing sizes and configurability. The most common type of logic cell is called Adaptive Logic Modules (ALM), these can be configured to perform any operation. Though ALM's can be configured to perform multiplication, this would be very inefficient and so FPGA's are generally equipped with Digital Signal Processing (DSP) slices which are specifically made to perform multiplication. Unfortunately, due to the die space requirements, there are fewer available. To illustrate, on any particular FPGA, ALM's are usually available in the order of tens of thousands, whereas there are usually only DSP's available in the order of tens. If a design's speed relies on multiplication, the limited number of DSP-slices could bottle-neck the computational speed.

The \textit{reduced cipher} implementation as presented in Section \ref{sec:fpga} reduces the computational load of the scheme on any platform, however one aspect is particularly beneficial to FPGA design. As shown in Table \ref{tab:hom_oper}, the total number of operations is reduced by an order of magnitude when using the \textit{reduced cipher}. More importantly however, is that all multiplication operations are replaced by bit-operations and additions. Replacing all multiplications with bit-operations ensures the FPGA design will not be bottle-necked by the availability of DSP-slices.

\subsection{Experimental Setup}
\noindent The encrypted control scheme has been implemented in VHDL for use on two Nexys 4 FPGA's in the configuration as shown in Figure \ref{fig:setup}. The results are obtained from a hardware simulation of the FPGA coupled with a high resolution simulation of the double pendulum. In the following first the choice for FPGA's as hardware platform is argued. Then, the simulated plant and the corresponding controller is described.

FPGA's can be programmed to operate without the need for a software layer and so is the platform chosen for implementation. Furthermore, it can be seen from table \ref{tab:hom_oper} that multiplication operations, which are the most computationally expensive on an FPGA, are completely eliminated by using the \textit{reduced cipher}. With this implementation on an FPGA a new control input can be generated every $0.8~ms$, which would not have been possible using the original ciphers or on conventional hardware.

The chosen plant is the inverted double pendulum depicted in Figure \ref{fig:rot_pend}. The dynamics of the double pendulum's state $\theta=[\theta_1~\theta_2]^\top$ is modeled as
\begin{equation}\label{eq:pend_model}
\hspace{-8pt}
\begin{aligned}
    &\left\{
\begin{aligned}
    &M(\theta)\ddot{\theta}+C(\theta,\dot{\theta})\dot{\theta}+G(\theta)=T\,,\\
    &T+\tau_e \dot{T}=k_m u\\
\end{aligned}\right.\\%
&\begin{aligned}
    M(\theta)&=\left[ \begin{matrix}P_1+P_2+2P_3\cos{\theta_2} & P_2+P_3\cos{\theta_2}\\ P_2+P_3\cos{\theta_2}&P_2\end{matrix}\right]\\
    C(\theta,\dot{\theta}) &= \left[ \begin{matrix}b_1-P_3\dot{\theta}_2\sin{\theta_2}&-P_3(\dot{\theta}_1+\dot{\theta}_2)\sin{\theta_2}\\P_3\dot{\theta}_1\sin{\theta_2}&b_2\end{matrix}\right]\\
    G(\theta)&=\left[ \begin{matrix}-g_1\sin{\theta_1}-g_2\sin{(\theta_1+\theta_2)}\\
    -g_2\sin{(\theta_1+\theta_2)}\end{matrix}\right]
\end{aligned}\\
&\begin{aligned}
    P_1=&m_1c_1^2+m_2l_1^2+I_1, \hspace{12pt}  P_2=m_2c_2^2+I_2\\
    P_3=&m_2l_1c_2,             \hspace{12pt}  g_1=(m_1c_1+m_2l_1)g, \hspace{12pt} g_2=m_2c_2g
\end{aligned}
\end{aligned}
\end{equation}
where $\theta_1$ and $\theta_2$ denote the angles of the pendulum links as shown in Figure \ref{fig:rot_pend}. The system has the same form as \eqref{eq:generalplant}. Furthermore, $m_1$,$m_2$ are the masses of the links; $l_1$, $l_2$ are their lengths; $c_1$, $c_2$ are the centers of mass; $I_1$, $I_2$ are the mass moments of inertia; $b_1$,$b_2$ are the damping coefficients of the joints; $k_m$,$\tau_e$ are the electrical motor gain and time constant, and $g$ is the gravitational acceleration.
\begin{figure}[H]
    \centering
    \includegraphics[width = 0.5\linewidth]{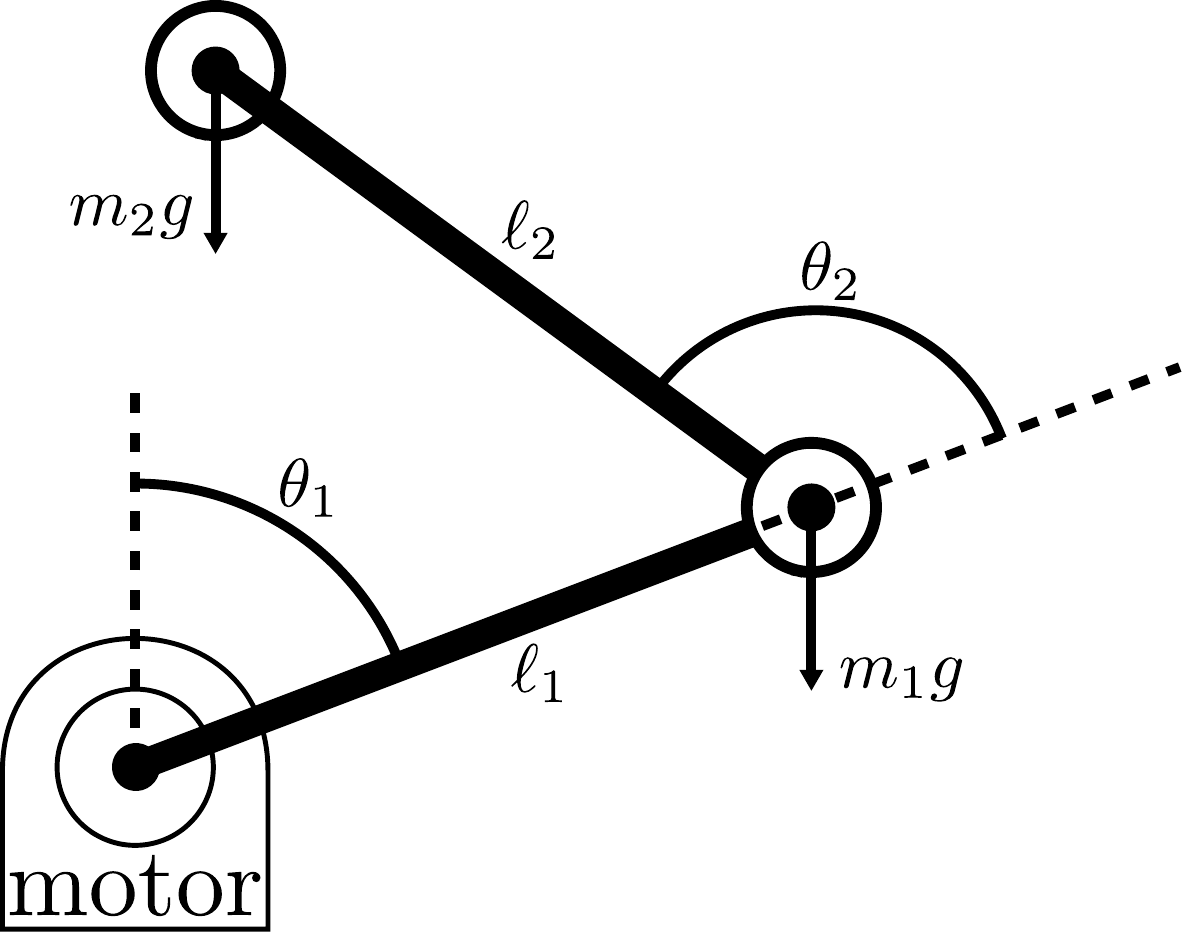}
    \caption{Double pendulum as modeled by Equation \eqref{eq:pend_model}.}
    \label{fig:rot_pend}
\end{figure}
\noindent The double pendulum is initialized at an initial state $\theta=\theta_0^\top = [0.0289, 0.1156]^\top$, $\dot{\theta}=\dot{\theta_0}^\top = [0.0669, 0.0049]^\top$ and $T = T_0 = 0$, and is controlled such that both pendulums point upwards, i.e. $\theta = [0~0]^\top$. A discrete time linearization of Model \eqref{eq:pend_model} can be made around $\theta = \dot{\theta} = [0~0]^\top$ as
\begin{equation}
\begin{aligned}
    \begin{dcases}
    x(k+1)=A_dx(k)+B_du(k)\,,\\
    y(k)=C_dx(k)\,,
\end{dcases}
\end{aligned}
\end{equation}
where $x(k)=[\theta_1(k)~\dot{\theta}_1(k)~\theta_2(k)~\dot{\theta}_2(k)~T]^\top$, $y(k)=[\theta_1(k)~\theta_2(k)]^\top$, and $A_d$, $B_d$, and $C_d$ are matrices of appropriate size. This linearized model is used to implement an observer and state feedback controller as
\begin{equation}\label{eq:pend_contrl}
    \hspace{-0.26cm}\begin{dcases}
        \hat{x}(k+1)=A_d\hat{x}(k)+B_du(k)+L(y(k)-C_d\hat{x}(k))\,, \\
        u(k+1)=K\hat{x}(k+1)\,,
    \end{dcases}
\end{equation}

\noindent where $L$ is the observer gain and $K$ is the state feedback gain. The controller takes the same form as \eqref{eq:generalcontroller}. The controller is updated at a rate of $f=100~Hz$. $L$ has been obtained by placing the observer poles at $[0.7~ 0.5~ 0.8~ 0.6~ 0.85]$ and $K=[-12.6~ -1.8~ -9.8~ -0.95~ 0.015]$. The input and the state estimate are initialized at $u(1)=0$ and $\hat{x}(0)=0$. The controller \eqref{eq:pend_contrl} is encrypted to obtain the equivalent controller $\tilde{g}(\cdot)$, $\tilde{v}(\cdot)$ of the form \eqref{eq:generalencryptedcontroller}. The model and encryption parameters used can be found in Table \ref{tab:params}.
\begin{table}[H]
    \centering
    \caption{Model and encryption parameters}
    \begin{tabular}{|c|c|c|c|}
    \hline
        Parameter & Value & Parameter & Value \\\hline
        $m_1$ & $0.125~kg$ & $m_2$ & $0.05~kg$\\
        $l_1$ & $0.1~m$ & $l_2$ & $0.1~m$\\
        $c_1$ & $-0.04~m$ & $c_2$ &$0.06~m$\\
        $I_1$ & $0.074~kgm^2$ & $I_2$ & $0.00012~kgm^2$\\
        $b_1$ & $4.8~kgs^{-1}$ & $b_2$ &$0.0002~kgs^{-1}$\\
        $k_m$ & $50~Nm$ & $\tau_e$ &$0.03~s$\\
        $g$ & $9.81ms^{-2}$ & $n$ & $7$\\
        $\ell$ & $64$ & $m$ & $7$\\
        $m_q$ & $10$ & $n_q$ & $22$\\
        $f$ & $100~Hz$ & &\\\hline
    \end{tabular}
    \label{tab:params}
\end{table}

The final control loop as shown in Figure \ref{fig:setup} works as follows: At boot-up FPGA 1, the adapter, generates an encryption key pair. FPGA 1 then encrypts the state space matrices and sends them to FPGA 2, the controller. Next the control loop starts. First, the adapter encrypts measurement vector $y$ and sends it to the controller \eqref{eq:generalencryptedcontroller} which computes $\textbf{E}(u^+)$ and $\textbf{E}(\hat{x}^+)$. To extend multiplicative depth and to prevent overflow, these signals are then sent to the adapter for decryption. This is a solution similar to that of \cite{HME_networked_control}. The control effort is applied to the plant, after which $u^+$ and $\hat{x}^+$ are bit-shifted, encrypted and sent back to the controller along with the new measurements $\mathbf{E}(y)$.

\subsection{Performance}
\noindent \noindent Figure \ref{fig:results} shows the results of system \eqref{eq:pend_model} being controlled according to $\tilde{g}(\cdot), \tilde{v}(\cdot)$. The encrypted observer estimates the states correctly and the plant is stabilized by the encrypted controller. Controlling the plant without encryption, i.e. according to \eqref{eq:pend_contrl}, yields identical results. This illustrates that the encrypted controller and unencrypted controller are indeed equivalent. 
\begin{figure}[H]
    \centering
    \includegraphics[width = 1\linewidth]{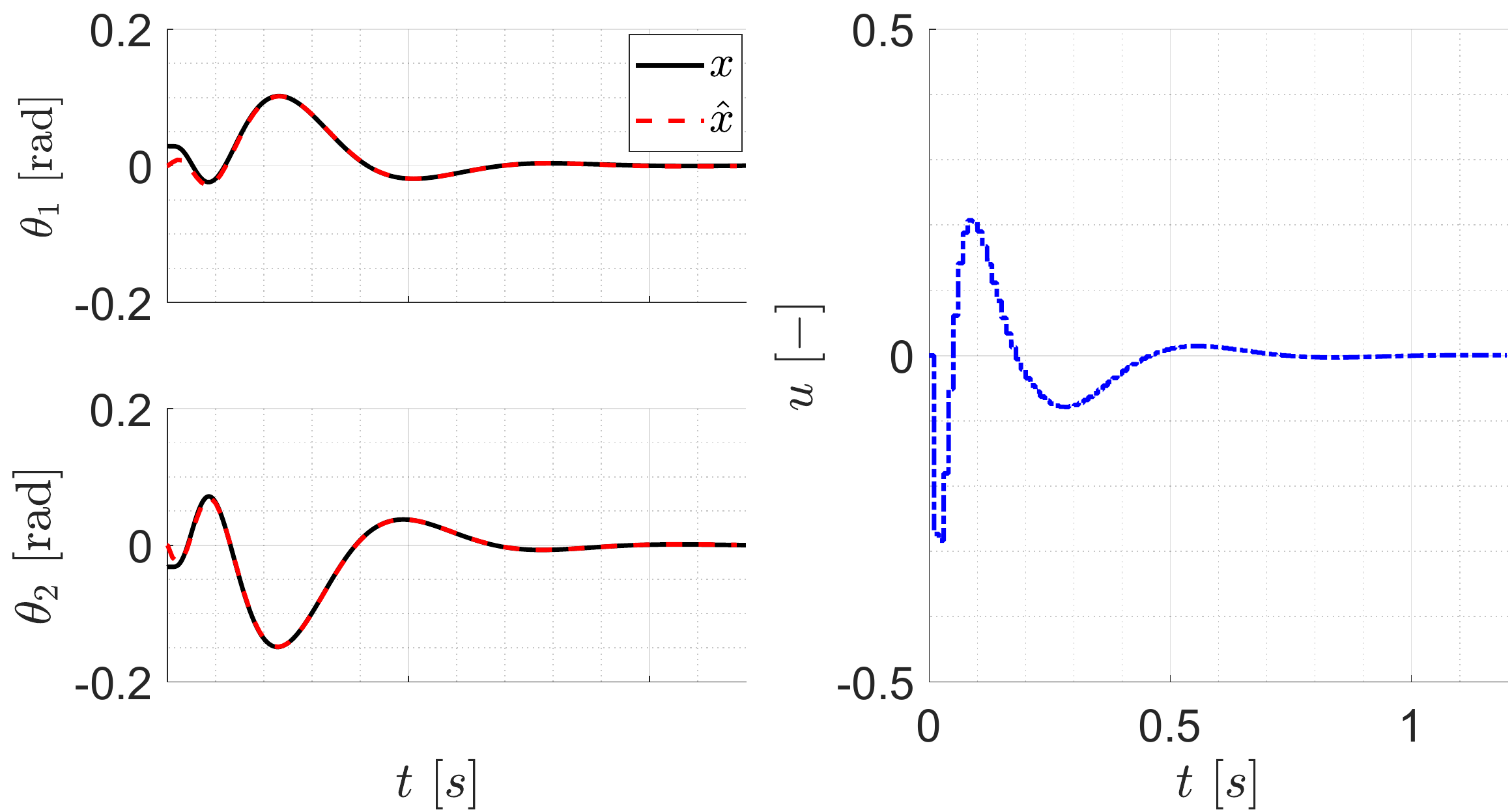}
    \caption{Simulation results, $\theta_1$, $\theta_2$ and control effort $u$}
    \label{fig:results}
\end{figure}
The experimental setup serves to highlight the contributions made to FHE. One can see that the plant is controlled towards an unstable equilibrium which requires a fast update rate of the encrypted controller. Due to the use of the \textit{reduced cipher}, this has become possible on the chosen hardware (Nexys 4 FPGA).

\section{Conclusion}\label{sec:concl}
\noindent The use of large scale systems such as hydroelectric dams or energy grids, has led to the need for secure monitoring and control over large distances. Securing such control systems from cyber-attacks is important to the safe operation. One of the ways to achieve this is through encryption. 

Using traditional encryption schemes, only the communication links can be secured, but signals have to be decrypted at the controller to calculate the control action. Fully Homomoprhic Encryption (FHE) has been developed such that operations can be performed on encrypted signals. Therefore, it has the potential to close the loop of encryption for secure control.
The main obstacle to widespread implementation of FHE in control is the high computational complexity. In this paper, the so-called \emph{reduced cipher} has been introduced, which allows for reducing the the computational complexity significantly. Specifically, the total number of operations performed is reduced by an order of magnitude. The \textit{reduced cipher} and analytical description of the encryption scheme are meant to enable more intuitive implementation and manipulation of the Gentry scheme for control purposes and extension of the capabilities of the scheme.

The presented FHE scheme is the first, to the best of the authors knowledge, that has been implemented for real-time control of an unstable plant. In future work we would like to extend the principle to more complex plants to show the full capability of the scheme. Furthermore, we will explore how to perform right hand bit shifts and other operations on encrypted data. This would be an elegant solution to the problem of shifting decimal points when multiplying fixed precision numbers and could enable to implement more complex control techniques.

\bibliographystyle{IEEEtran}
\bibliography{biblio}

\begin{thebibliography}{10}
\providecommand{\url}[1]{#1}
\csname url@samestyle\endcsname
\providecommand{\newblock}{\relax}
\providecommand{\bibinfo}[2]{#2}
\providecommand{\BIBentrySTDinterwordspacing}{\spaceskip=0pt\relax}
\providecommand{\BIBentryALTinterwordstretchfactor}{4}
\providecommand{\BIBentryALTinterwordspacing}{\spaceskip=\fontdimen2\font plus
\BIBentryALTinterwordstretchfactor\fontdimen3\font minus
  \fontdimen4\font\relax}
\providecommand{\BIBforeignlanguage}[2]{{%
\expandafter\ifx\csname l@#1\endcsname\relax
\typeout{** WARNING: IEEEtran.bst: No hyphenation pattern has been}%
\typeout{** loaded for the language `#1'. Using the pattern for}%
\typeout{** the default language instead.}%
\else
\language=\csname l@#1\endcsname
\fi
#2}}
\providecommand{\BIBdecl}{\relax}
\BIBdecl

\bibitem{AES}
M.~Dworkin, E.~Barker, J.~Nechvatal, J.~Foti, L.~Bassham, E.~Roback, and
  J.~Dray, ``\BIBforeignlanguage{en}{Advanced encryption standard (aes)},''
  2001-11-26 2001.

\bibitem{RSA}
R.~L. Rivest, A.~Shamir, and L.~Adleman, ``A method for obtaining digital
  signatures and public-key cryptosystems,'' \emph{Communications of the ACM},
  vol.~21, no.~2, p. 120–126, 1978.

\bibitem{EncInCrtl}
R.~Smith, ``Cryptography concepts and effects on control system
  communications,'' 2018.

\bibitem{ElGamal}
T.~El~Gamal, ``A public key cryptosystem and a signature scheme based on
  discrete logarithms,'' in \emph{Advances in Cryptology}, G.~R. Blakley and
  D.~Chaum, Eds.\hskip 1em plus 0.5em minus 0.4em\relax Berlin, Heidelberg:
  Springer, 1985, pp. 10--18.

\bibitem{Paillier}
P.~Paillier, ``Public-key cryptosystems based on composite degree residuosity
  classes,'' \emph{EUROCRYPT}, p. 223–238, 1999.

\bibitem{Gentry2009}
C.~Gentry, ``A fully homomorphic encryption scheme,'' Ph.D. dissertation,
  Stanford University, 2009.

\bibitem{gentry2013}
C.~Gentry, A.~Sahai, and B.~Waters, ``Homomorphic encryption from learning with
  errors: Conceptually-simpler, asymptotically-faster, attribute-based,''
  \emph{Advances in Cryptology}, p. 75–92, 2013.

\bibitem{Cheon2015}
J.~H. Cheon and D.~Stehl{\'e}, ``Fully homomophic encryption over the integers
  revisited,'' in \emph{Advances in Cryptology -- EUROCRYPT 2015}, E.~Oswald
  and M.~Fischlin, Eds.\hskip 1em plus 0.5em minus 0.4em\relax Berlin,
  Heidelberg: Springer Berlin Heidelberg, 2015, pp. 513--536.

\bibitem{HME_networked_control}
K.~Kogiso and T.~Fujita, ``Cyber-security enhancement of networked control
  systems using homomorphic encryption,'' in \emph{CDC}, 12 2015.

\bibitem{RTS_HME_FPGA}
J.~Tran, F.~Farokhi, M.~Cantoni, and I.~Shames, ``Implementing homomorphic
  encryption based secure feedback control,'' \emph{Control Engineering
  Practice}, vol.~97, p. 104350, 2020.

\bibitem{PeriodicReset}
C.~Murguia, F.~Farokhi, and I.~Shames, ``Secure and private implementation of
  dynamic controllers using semi-homomorphic encryption,'' 2019.

\bibitem{Kim2016}
J.~Kim, C.~Lee, H.~Shim, J.~H. Cheon, A.~Kim, M.~Kim, and Y.~Song, ``Encrypting
  controller using fully homomorphic encryption for security of cyber-physical
  systems,'' \emph{IFAC-PapersOnLine}, vol.~49, no.~22, pp. 175--180, 2016, 6th
  IFAC Workshop on Distributed Estimation and Control in Networked Systems
  NECSYS 2016.

\bibitem{Quant_proj}
J.~Kim, H.~Shim, and K.~Han, ``Dynamic controller that operates over
  homomorphically encrypted data for infinite time horizon,'' 2019.

\bibitem{Kim2021}
J.~Kim, H.~Shim, H.~Sandberg, and K.~H. Johansson, ``{Method for Running
  Dynamic Systems over Encrypted Data for Infinite Time Horizon without
  Bootstrapping and Re-encryption},'' in \emph{60th IEEE Conference on Decision
  and Control}, 2021, pp. 5614--5619.

\bibitem{Chaher2021}
M.~P. Chaher, B.~Jayawardhana, and J.~Kim, ``{Homomorphic Encryption-Enabled
  Distance-Based Distributed Formation Control with Distance Mismatch
  Estimators},'' in \emph{60th IEEE Conference on Decision and Control}, 2021,
  pp. 4915--4922.

\bibitem{Cheon2018}
J.~Cheon, K.~Han, S.-M. Hong, H.~Kim, J.~Kim, S.~Kim, H.~Seo, H.~Shim, and
  Y.~Song, ``Toward a secure drone system: Flying with real-time homomorphic
  authenticated encryption,'' \emph{IEEE Access}, vol.~6, pp. 24\,325--24\,339,
  2018.

\bibitem{control_theory_introduction_book}
K.~J. {\AA}str\"om and R.~M. Murray, \emph{Feedback systems: An introduction
  for scientists and Engineers}.\hskip 1em plus 0.5em minus 0.4em\relax
  Princeton University Press, 2021.

\bibitem{LWE_oded_regev}
O.~Regev, ``On lattices, learning with errors, random linear codes, and
  cryptography,'' \emph{J. {ACM}}, vol.~56, no.~6, pp. 34:1--34:40, 2009.

\bibitem{MPDec}
D.~Micciancio and C.~Peikert, ``Trapdoors for lattices: Simpler, tighter,
  faster, smaller,'' in \emph{Advances in Cryptology -- EUROCRYPT 2012},
  D.~Pointcheval and T.~Johansson, Eds.\hskip 1em plus 0.5em minus 0.4em\relax
  Berlin, Heidelberg: Springer Berlin Heidelberg, 2012, pp. 700--718.

\bibitem{FHE_Float}
S.~Moon and Y.~Lee, ``An efficient encrypted floating-point representation
  using {HEAAN} and {TFHE},'' \emph{Security and Communication Networks}, vol.
  2020, pp. 1--18, 03 2020.

\end{thebibliography}

\end{document}